\documentclass{ifacconf}

\usepackage{graphicx}      
\usepackage{natbib}        
\usepackage{amsfonts}
\usepackage{algorithm}
\usepackage{algpseudocode}
\usepackage{amsmath}
\usepackage{cases}
\usepackage[dvipsnames]{xcolor}
\usepackage[most]{tcolorbox}
\usepackage{graphicx}
\usepackage{subcaption}
\usepackage{booktabs}
\newtheorem{ass}{Assumption} 
\newenvironment{proof}[1][Proof]{\par\noindent\textit{#1.} }{\hfill$\square$\par}
\begin{document}
\begin{frontmatter}

\title{Efficient Uniform Feasible-Set Sampling for Approximate Linear MPC} 


\author[First,Second]{Elias Milios} 
\author[First]{Felix Berkel}
\author[First]{Felix Gruber}
\author[Second]{Melanie N. Zeilinger} 
\author[First]{Kim P. Wabersich}

\address[First]{Corporate Research of Robert Bosch GmbH, Germany (e-mail:  \{{elias.milios, felix.berkel, felix.gruber, kimpeter.wabersich\}@de.bosch.com})}
\address[Second]{Institute of Dynamic Systems and Control, ETH Zurich, Switzerland (e-mail: \{emilios,mzeilinger\}@ethz.ch)}

\begin{abstract}
    Model Predictive Control (MPC) offers safe and near-optimal control but suffers from high computational costs. Approximate MPC (AMPC) mitigates this by learning a cheaper surrogate policy, typically by training a neural network on state-MPC input pairs. Generating training data is a major bottleneck, requiring solving the MPC for numerous states sampled from its feasible set. Since this feasible set is implicitly defined and unknown, efficient sampling is nontrivial but crucial. We propose the linear MPC Hit-and-Run (LMPC-HR) sampler for linear MPC with polyhedral constraints. We identify the feasible set boundaries along search directions, a crucial step within HR, by formulating the problem as a convex linear program, replacing expensive iterative searches with a single optimization step. A numerical study demonstrates that LMPC-HR reduces the computational cost of generating uniformly distributed samples from the feasible set by an order of magnitude compared to standard baselines.
\end{abstract}


\end{frontmatter}

\section{Introduction}
%
MPC offers near-optimal performance and strict safety for modern control systems by repeatedly solving an optimal control problem online. However, its widespread industrial application is hindered by the significant engineering effort and hardware costs required to guarantee real-time embedded deployment. AMPC aims to overcome this by approximating the MPC policy with a computationally efficient surrogate policy, typically a neural network. This enables real-time deployment on embedded hardware with low and predictable compute and memory demands~\citep{AMPC_Berkel, AMPC_adaptive_Hose}.\par
%
The standard process for deriving a surrogate policy relies on training a neural network via supervised learning based on a dataset containing state-MPC label pairs. Generating the dataset involves sampling initial states from a user-defined state distribution and solving the MPC problem offline for each state to retrieve the corresponding MPC label. The choice of the MPC label is dictated by the specific training procedure. For example, behavioral cloning and improved variants focus on the MPC input and gradients thereof~\citep{AMPC_Hertneck, AMPC_Lüken}, while more advanced learning procedures may also employ the optimal value function~\citep{AMPC_Milios}.\par
Crucially, the performance of the resulting surrogate policy highly depends on the training data's sample size and distribution, often leading to poor performance in areas of the state space that are not sufficiently covered~\citep{DAgger_Ross}. To mitigate this, the uniform distribution often serves as a natural baseline, particularly in settings where the state distribution encountered during deployment is unknown, ensuring a broad coverage of the state space. Generating suitable data reflecting a uniform distribution, however, constitutes a major challenge for two main reasons. First, it necessitates solving the computationally expensive MPC problem for all desired states, which can easily number in the millions~\citep{AMPC_safety_Hose_Köhler}. Second, valid training samples must be drawn from the feasible set of the MPC, which is implicitly defined by the solvability of the MPC problem and therefore typically unknown.\par
To this end, common sampling strategies for generating uniformly distributed data typically involve rejection-based approaches. Here, samples are drawn from a uniform distribution over a bounding box or sphere containing the state constraints, or from a uniform grid discretizing such a bounding set, and rejected if the MPC solver detects infeasibility. While simple to implement, these approaches suffer from the curse of dimensionality, scaling exponentially in the state dimension. Moreover, since the feasible set of the MPC is unknown and often a small subset of such bounding sets, these naive approaches can result in excessive rejection rates. The high rejection rates constitute an additional computational burden, as MPC solvers often require many iterations to detect infeasibility. Together, this leads to significant computational overhead, often rendering the data generation process a computational bottleneck in AMPC.\par
A theoretically grounded alternative that addresses uniform sampling from implicitly defined sets is the HR Markov Chain Monte Carlo sampler~\citep{Hit-and-Run_Belisle}. Unlike naive rejection‑based approaches, HR scales favorably to higher dimensions, achieving polynomial complexity for convex bodies~\citep{Hit-and-Run_Lovasz}. The algorithm constructs a Markov chain that iteratively explores the interior of the feasible set, reducing the high‑dimensional sampling task at each step to a one‑dimensional sampling problem along a line segment contained within the feasible set. Since the feasible set is defined implicitly by the solvability of the MPC problem, identifying the intersection points of the search line with the feasible set boundary, however, creates a critical computational bottleneck. In particular, estimating these boundary points with standard strategies such as geometric bisection~\citep{Hit-and-Run_Rudolf}, incremental stepping~\citep{AMPC_Chen}, or uniform sampling along the search line~\citep{Hit-and-Run_improved_Smith}, often necessitates solving the MPC problem numerous times.\par
\textit{Contribution:}
This work proposes the LMPC-HR sampler, enabling the efficient generation of samples whose stationary distribution is approximately uniform over the feasible set of linear MPC with polyhedral constraints. We build upon the HR framework, enhancing its efficiency in the context of AMPC by framing the problem of identifying the intersection points of a search line with the feasible set as a convex linear program (LP). This enables the efficient calculation of such a boundary point with a single global optimization step, circumventing the computational burden of iterative methods. Our specific contributions are:
\begin{itemize} 
\item The LMPC-HR sampler, a sampling strategy tailored for linear AMPC, enabling the efficient generation of guaranteed uniformly distributed training data.
\item The derivation of a convex LP formulation that exactly determines the intersection of a search line with the boundary of the implicit feasible set of the MPC, which constitutes the core feature of the proposed LMPC-HR sampler.
\item A numerical study comparing the LMPC-HR sampler against standard sampling baselines, including uniform rejection sampling and naive HR approaches. The study demonstrates an order of magnitude reduction in computation time for generating uniformly distributed state samples within the MPC feasible set.
\end{itemize}
%
%
\textit{Related Work:}
Uniform sampling is a widely adopted strategy in AMPC, compare, e.g.,  with~\cite{AMPC_Hertneck, AMPC_uniform_sampling_Krishnamoorthy,AMPC_Drgona}, yet the question of how to sample the feasible set efficiently remains under-addressed. To this end,~\cite{AMPC_Chen} applies a random walk strategy within the feasible set of convex linear MPC problems. The approach relies on incremental stepping within fixed intervals to traverse the feasible set along a search direction. For each incremental step, the MPC problem is solved, and if feasible, the resulting sample is collected. While this ensures valid sampling from the interior of the feasible set, it does not ensure uniformity of the generated samples. \cite{AMPC_Winqvist} proposed using HR sampling for convex linear AMPC to improve the efficiency of the sampling process. However, the approach relies on pre-computing the feasible set as a convex maximal control invariant set to obtain an explicit polyhedral description. This, however, can be computationally prohibitive, especially for high-dimensional systems. Our work advances this paradigm by formalizing the boundary identification directly as an optimization problem. By leveraging the specific structure of linear MPC with polyhedral constraints, we extend the HR framework to high-dimensional systems where obtaining explicit geometric representations is computationally intractable, while avoiding the computational cost of iterative search procedures.\par
\textit{Outline:}
The paper is structured as follows. Section~\ref{sec: preliminaries} introduces the problem setting. Section~\ref{sec: boundary id section} presents the LMPC-HR sampler, where we first derive the LP-based boundary identification in Section~\ref{sec: LP boundary oracle}, then present the overall algorithm in Section~\ref{sec: lin Algo}, and finally analyze its complexity in Section~\ref{sec: lin complexity discussion}. Section~\ref{sec: experiments} provides a numerical evaluation, and Section~\ref{sec: conclusion} concludes the paper.
\section{Problem Setting}\label{sec: preliminaries}
We consider linear discrete-time systems of the form
\begin{equation}\label{eq: linear system}
    x(k+1) = Ax(k) + Bu(k),
\end{equation}
where $x(k) \in \mathbb{R}^{n_x}$ and $u(k) \in \mathbb{R}^{n_u}$ are the system state and control input, respectively, $A \in \mathbb{R}^{n_x \times n_x}$ and $B\in \mathbb{R}^{n_x \times n_u}$ denote the system matrices, and $k \in \mathbb{N}$ denotes the discrete time index.\par 
\subsection{Model Predictive Control}\label{sec: MPC}
Based on system~\eqref{eq: linear system}, we consider the following MPC problem
\begin{subequations}\label{eq: linear MPC}
\begin{align}
    V_\mathrm{MPC}(x(k)) & :=  \min_{u_{\cdot|k},\: x_{\cdot|k}} \ \sum_{i=0}^{N-1} \ell(x_{i|k},u_{i|k}) + V_\mathrm{f}(x_{N|k})\label{eq: generic MPC cost}\\
    \text{s.t. } \  & x_{0|k} = x(k), \quad x_{N|k} \in \mathbb{X}_\mathrm{f}, \label{eq: init and term constraint}\\
    & \forall i \in \{0,1,\ldots,N-1\}:\notag \\
    & \quad x_{i+1|k} = Ax_{i|k} + Bu_{i|k}, \label{eq: dynamics constraint}\\
    & \quad u_{i|k} \in \mathbb{U}, \quad x_{i|k} \in \mathbb{X},\label{eq: state input constraints}
\end{align}
\end{subequations}
where we optimize over a predicted input and state sequence $u_{\cdot|k} = \{u_{i|k}\}_{i=0}^{N-1}$ and $x_{\cdot|k} = \{x_{i|k}\}_{i=0}^{N}$, respectively, with $i$ denoting the prediction time step and $N$ denoting the prediction horizon. The cost~\eqref{eq: generic MPC cost} is defined by the stage cost $\ell:\mathbb{R}^{n_x} \times \mathbb{R}^{n_u} \rightarrow \mathbb{R}_{\geq 0}$ and terminal cost $V_\mathrm{f}:\mathbb{R}^{n_x} \rightarrow \mathbb{R}_{\geq 0}$. Equation~\eqref{eq: state input constraints} implements state and input constraints of the form $\mathbb{X}:=\{x \in \mathbb{R}^{n_x} \mid H_x x \leq h_x\}$ and $\mathbb{U}:=\{u \in \mathbb{R}^{n_u} \mid H_uu\leq h_u \}$, where $H_x \in \mathbb{R}^{n_{c,x} \times n_x}$, $h_x \in \mathbb{R}^{n_{c,x}}$, $H_u \in \mathbb{R}^{n_{c,u} \times n_u}$, and $h_u \in \mathbb{R}^{n_{c,u}}$ are constant matrices and vectors. Similarly, Equation~\eqref{eq: init and term constraint} implements a terminal state constraint of the form $\mathbb{X}_\mathrm{f}:=\{x \in \mathbb{R}^{n_x} \mid H_\mathrm{f} x \leq h_\mathrm{f}\}$ with constant matrix $H_\mathrm{f}\in \mathbb{R}^{n_{c,\mathrm{f}} \times n_x}$ and vector $h_\mathrm{f} \in \mathbb{R}^{n_{c, \mathrm{f}}}$.\par
The feasible set $\mathcal{X}_N$ of~\eqref{eq: linear MPC} is given as the set of all states $x(k)$ for which the optimization problem admits a feasible solution, formally defined by
\begin{equation}\label{eq: feasible set}
    \mathcal{X}_N:=\{x(k) \in \mathbb{R}^{n_x} \mid \exists u_{\cdot|k} \in \mathbb{R}^{Nn_u}: \eqref{eq: init and term constraint} - \eqref{eq: state input constraints}\}.
\end{equation} 
\begin{ass}\label{ass: boundedness}
    We assume that $\mathcal{X}_N$ is bounded and has a nonempty interior.
\end{ass}
If $\mathbb{X}$, $\mathbb{U}$, and $\mathbb{X}_\mathrm{f}$ are bounded, then boundedness of $\mathcal{X}_N$ follows by definition. Moreover, for system dynamics typically encountered in practice, bounded input constraints and a bounded terminal set often result in a bounded feasible set, even if the state constraints are unbounded, see, e.g.,~\cite{MPC_Rawlings}.\par
The optimal solution to~\eqref{eq: linear MPC} for any $x(k) \in \mathcal{X}_N$ is denoted by $u^*_{\cdot|k}(x(k))=\{u^*_{i|k} (x(k))\}_{i=0}^{N-1}$, and the resulting MPC input is given by its first element $u(k)=\pi_\mathrm{MPC}(x(k)):=u^*_{0|k}(x(k))$.
\subsection{Problem Formulation}\label{sec: problem formulation}
The overarching goal of this work is the efficient generation of training datasets for AMPC. Specifically, we aim to collect a dataset $\mathcal{D}=\{ s(x_j) \}_{j=1}^{N_s}$ containing $N_s$ state-MPC feature tuples $s(x)$ uniformly distributed over the implicitly defined feasible set $\mathcal{X}_N$. The structure of the tuple $s(x)$ depends on the specific imitation learning scheme, and, in the most common setting, collects state-input pairs, implying $s(x) = (x,\pi_\mathrm{MPC}(x))$.\par
To this end, we utilize the HR sampler, which iteratively constructs a Markov chain $\{x_j\}_{j=1}^{N_s}$ as uniform samples from the feasible set $\mathcal{X}_N$. Based on an initial feasible state $x_1 \in \mathcal{X}_N$, the method iterates as follows:
\begin{enumerate}
    \item[1.] At state $x_j$, sample $d_j$ uniformly from the L2-Ball $\{d \in \mathbb{R}^{n_x} \mid \lVert d \rVert \leq 1\}$ and set $d_j = \frac{d_j}{\lVert d_j\rVert}$.
    \item[2.] Identify the feasible line segment 
    \begin{equation}\label{eq:feasible_line_seg}
        \mathcal{S}_j := \{ x_j \pm \alpha d_j \mid \alpha \geq 0\} \cap \mathcal{X}_N.
    \end{equation}
    \item[3.] Sample $x_{j+1}$ uniformly from $\mathcal{S}_j$.
    \item[4.] If $j < N_s+1$: set $j\leftarrow j+1$  and go to 1.
\end{enumerate}\par
The computational challenge is the identification of the line segment $\mathcal{S}_j$ in Step~2, which requires determining the boundary limits $\alpha^*_-$ and $\alpha^*_+$  such that $\mathcal{S}_j = \{x_j + \alpha d_j \mid \alpha \in [-\alpha^*_-, \alpha^*_+] \}$.\par
Existing solutions typically rely on iterative search procedures, such as bisection or incremental stepping~\citep{Hit-and-Run_Rudolf, AMPC_Chen}. These methods treat the underlying MPC problem as a black box, requiring repeated feasibility checks to approximate the boundary. Since each check entails solving the optimization problem~\eqref{eq: linear MPC}, this iterative process becomes the dominant computational bottleneck.\par
In contrast, we propose to formulate the problem of identifying the boundary limits $\alpha^*_-$ and $\alpha^*_+$, referred to as the \textit{boundary identification problem} in the following, as an optimization problem. With this, instead of estimating $\alpha^*_-$ and $\alpha^*_+$ iteratively, we directly obtain these boundary limits by maximizing the step size $\alpha$ subject to the feasibility of~\eqref{eq: linear MPC}. Conceptually, for a given $x \in \mathcal{X}_N$, this optimization takes the form
\begin{equation}\label{eq: abstract_boundary_opt}
\Phi(x,d) =\max_{\alpha \geq 0} \ \alpha \quad \text{s.t.} \quad x + \alpha d \in \mathcal{X}_N.
\end{equation}
For linear systems with polyhedral constraints, this is a convex LP, as detailed in the upcoming section. Solving~\eqref{eq: abstract_boundary_opt} for the directions $d_j$ and $-d_j$ yields the exact boundary limits $\alpha^*_+ = \Phi(x_j, d_j)$ and $\alpha^*_- = \Phi(x_j, -d_j)$ in a single optimization step per direction.\par
\section{Efficient Boundary Identification}\label{sec: boundary id section}
This section presents the proposed LMPC-HR sampler, which utilizes a tractable LP formulation of the boundary identification problem~\eqref{eq: abstract_boundary_opt}. Section~\ref{sec: LP boundary oracle} establishes this LP formulation. Section~\ref{sec: lin Algo} details its integration into the HR framework, yielding the LMPC-HR sampler summarized in Algorithm~\ref{alg: lmpc_hr}. Section~\ref{sec: lin complexity discussion} analyzes the computational complexity compared to baseline methods.\par
\subsection{LP Formulation of Boundary Identification Problem}\label{sec: LP boundary oracle}
To solve~\eqref{eq: abstract_boundary_opt}, we must find the maximum step size $\alpha$ such that a limiting inequality constraint becomes active while maintaining the system dynamics constraints. Exploiting the linearity of the system~\eqref{eq: linear system}, we eliminate the dynamics constraints by parameterizing the predicted trajectory through the control sequence, reducing the problem to linear inequalities in the control inputs.\par
To this end, we define $z_u:= [u_{0|k}^\top, \dots, u_{N-1|k}^\top]^\top$ as the stacked sequence of predicted control inputs. By recursively substituting the linear dynamics~\eqref{eq: linear system}, the vector of predicted states $z_x := [x_{1|k}^\top, \dots, x_{N|k}^\top]^\top$ can be expressed as an affine function of the initial state $x(k)$ and the input sequence $z_u$:
\begin{equation}\label{eq: state_prediction}
    z_x = \Omega x(k) + \Gamma z_u,
\end{equation}
where the prediction matrices $\Omega \in \mathbb{R}^{N n_x \times n_x}$ and $\Gamma \in \mathbb{R}^{N n_x \times N n_u}$ are defined as
\begin{equation*}
\Omega = \begin{bmatrix} A \\ A^2 \\ \vdots \\ A^N \end{bmatrix}, \quad 
\Gamma = \begin{bmatrix}
    B & 0 & \ldots & 0 \\ 
    AB & B & \ldots & 0 \\
    \vdots & \vdots & \ddots & \vdots \\
    A^{N-1}B & A^{N-2}B & \ldots & B
\end{bmatrix}.
\end{equation*}
It is well known that this enables expressing the constraints in~\eqref{eq: linear MPC} solely in terms of the decision variable $z_u$ and the parameter $x(k)$, compare with~\cite{Explicit_MPC_Bemporad}. Specifically, the constraints~\eqref{eq: init and term constraint}--\eqref{eq: state input constraints} can be rewritten as
\begin{equation}\label{eq: compact_constraints}
    G z_u \leq w + F x(k).
\end{equation}
Here, $w \in \mathbb{R}^{n_c}$ is the stacked vector of the constraint vectors $h_u, h_x$, and $h_\mathrm{f}$, and the matrices $G \in \mathbb{R}^{n_c \times N n_u}$ and $F \in \mathbb{R}^{n_c \times n_x}$ are composed of the constraint matrices $H_u, H_x$ and $H_\mathrm{f}$ projected through the matrices $\Omega$ and $\Gamma$. The total number of linear inequality constraints is given by $n_c = N(n_{c,u} + n_{c,x}) + n_{c,\mathrm{f}}$.\par
Since~\eqref{eq: compact_constraints} parameterizes all inequality constraints subject to dynamic feasibility, we can substitute the search line $x(k) = x_j + \alpha d_j$ into~\eqref{eq: compact_constraints} and maximize over $\alpha$ to find the maximum feasible step size. This yields the following convex LP 
\begin{subequations}\label{eq: lp_oracle_final}
\begin{align}
    \Phi(x_j, d_j) &:= \max_{\alpha, z_u} \quad  \alpha \\
    \text{s.t.}\quad  G z_u - \alpha F d_j & \leq w + F x_j, \\
    \alpha &\geq 0.
\end{align}
\end{subequations}
Problem~\eqref{eq: lp_oracle_final} is an LP with $N n_u + 1$ decision variables and $n_c$ constraints, comparable in size to~\eqref{eq: linear MPC}, with only one additional variable $\alpha\in \mathbb{R}$. Standard LP solvers can compute the exact boundary limits $\alpha^*_-$ and $\alpha^*_+$ efficiently per direction, enabling the construction of $\mathcal{S}_j=\{x_j + \alpha d_j \mid \alpha \in [-\alpha^*_-, \alpha^*_+]\}$ without iterative feasibility checks. Moreover, problem~\eqref{eq: lp_oracle_final} remains a convex LP regardless of the MPC cost terms in~\eqref{eq: generic MPC cost} since it addresses only feasibility. 
\subsection{LMPC-HR Sampler}\label{sec: lin Algo}
This section integrates problem~\eqref{eq: lp_oracle_final} into the HR framework, yielding the LMPC-HR sampler in Algorithm~\ref{alg: lmpc_hr}.\par
Given an initial feasible state $x_1\in\mathcal{X}_N$ (e.g., the origin) and a dataset size $N_s$, each iteration $j=1,\ldots,N_s$ follows these steps. Sample a unit direction $d_j$. Solve the LP~\eqref{eq: lp_oracle_final} for $(x_j,d_j)$ and $(x_j,-d_j)$ to obtain the boundary limits $\alpha^*_+$ and $\alpha^*_-$, respectively. Draw $x_{j+1}$ uniformly from the feasible segment $\mathcal{S}_j=\{x_j+\alpha d_j\mid\alpha\in[-\alpha^*_-,\alpha^*_+]\}$. Record the training tuple $s(x_j)$ (state and MPC feature) in $\mathcal{D}$.\par
The following lemma verifies that the probability distribution of the samples produced by the LMPC-HR sampler converges to the uniform distribution over $\mathcal{X}_N$.
\begin{lem}\label{lem: convergence}
Suppose Assumption~\ref{ass: boundedness} holds. Then, the probability distribution of the state $x_j$ generated by Algorithm~\ref{alg: lmpc_hr} converges in total variation to the uniform distribution over $\mathcal{X}_N$ as $j\rightarrow \infty$.
\end{lem}
\begin{proof}
Because the equality constraints in~\eqref{eq: linear MPC} are linear and the inequality constraints in~\eqref{eq: linear MPC} define a convex polyhedron, $\mathcal{X}_N$ is a convex polyhedron, see~\cite{Explicit_MPC_Bemporad}. Furthermore, since $\mathcal{X}_N$ is bounded and has nonempty interior by Assumption~\ref{ass: boundedness}, we have that $\mathcal{X}_N$ is a convex body. Therefore,~\cite[Theorem~3]{Hit-and-Run_Lovasz} establishes that the probability distribution of $x_j$ generated by Algorithm~\ref{alg: lmpc_hr} converges in total variation to the uniform distribution over $\mathcal{X}_N$ as $j\rightarrow\infty$.
\end{proof}\par
\begin{algorithm}[t]\caption{LMPC-HR: Linear MPC HR Sampler}\label{alg: lmpc_hr}
\begin{algorithmic}[1]
\Statex \textbf{Input:} Initial state $x_1 \in \mathcal{X}_N$, Dataset size $N_s$.
\Statex \textbf{Output:} Dataset $\mathcal{D} = \{s(x_j)\}_{j=1}^{N_s}$.
\For{$j = 1$ to $N_s$}
\Statex \textbf{Data Collection:}
\State Solve MPC~\eqref{eq: linear MPC} for $x_j$ to retrieve $s(x_j)$.
\State Store $s(x_j)$ in $\mathcal{D}$.
\Statex \textbf{Sampling Step:}
\State Sample $d_j$ uniformly from the L2-Ball.
\State $\alpha^*_+ \leftarrow\Phi(x_j, d_j)$.
\State $\alpha^*_- \leftarrow \Phi(x_j, -d_j)$.
\Statex \textbf{Update:}
\State Sample step size $\beta$ uniformly from $[-\alpha^*_-, \alpha^*_+]$.
\State $x_{j+1} \leftarrow x_j + \beta d_j$.
\EndFor
\end{algorithmic}\end{algorithm}
%
%
%
Since $\mathcal{X}_N$ is a convex body, the LMPC-HR sampler also inherits the polynomial-time bound on the mixing time established in \citep[Theorem 3]{Hit-and-Run_Lovasz}. That is, the number of samples required to guarantee that the generated distribution is arbitrarily close to the uniform distribution scales polynomially with the state dimension.\par
We note that the computational cost required for solving the LP in lines~5 and~6 can be further reduced by warm-starting the solver using the solution from the previous iteration. Moreover, similar to~\cite{NMPC_Zavala_Biegler}, the sensitivities of the MPC solution, i.e., the gradient information of the optimal solution $u^*_{\cdot|k}(x_j)$ to~\eqref{eq: linear MPC} for a given state $x_j$, can be leveraged to provide a first-order approximation of the optimal solution  $u^*_{\cdot|k}(x_{j+1})$, serving as an informative warm-start for solving~\eqref{eq: linear MPC} at $x_{j+1}$ in line~2.
\subsection{Complexity Discussion}\label{sec: lin complexity discussion}
The primary benefit of the proposed strategy is that obtaining a feasible sample $x_j \in \mathcal{X}_N$ requires solving a convex LP exactly twice, instead of iteratively solving the MPC problem numerous times. Consequently, the computational advantage compared to alternative approaches is dominated by the effective number of MPC queries required to obtain a feasible sample $x_j \in \mathcal{X}_N$. In the following, we quantify this computational advantage by estimating the expected number of MPC queries required for baseline methods like volumetric rejection sampling, and HR approaches deploying directional rejection sampling and geometric bisection for the boundary identification. To facilitate this, we define the outer bounding box $\mathcal{B}(\mathbb{X}):=\{x \in \mathbb{R}^{n_x} \mid \lVert x \rVert_\infty \leq L\} \supseteq \mathbb{X}$ with side length $2L$ and $L:=\sup_{x \in \mathbb{X}} \lVert x \rVert$ as the worst-case uncertainty set in which any sampling and boundary identification method must operate.\par
\textit{Uniform Volumetric Rejection Sampling:} Unlike the subsequent methods, which operate as subroutines within the HR framework, volumetric rejection sampling is a standalone sampling strategy. It proceeds by solving~\eqref{eq: linear MPC} for samples directly drawn uniformly from the full-dimensional bounding box $\mathcal{B}(\mathbb{X})$, and rejecting those that are infeasible. The probability of drawing a sample from $\mathcal{X}_N$ is given by the ratio of volumes $\text{Vol}(\mathcal{X}_N) / \text{Vol}(\mathcal{B}(\mathbb{X}))$, where $\text{Vol}(\mathcal{B}(\mathbb{X}))$ scales exponentially with $(2L)^{n_x}$. The expected number of MPC queries scales inversely with this ratio, often being prohibitively large for higher-dimensional and constrained systems.\par
\textit{Directional Rejection Sampling-Based HR:} This method applies the same rejection principle as uniform volumetric rejection sampling within the HR framework. This restricts the sampling domain to the one-dimensional intersection of the search line and the bounding box, whose worst-case length is defined by the main diagonal of $\mathcal{B}(\mathbb{X})$ which scales as $2L\sqrt{n}$. By reducing the search space dimensionality from $n_x$ to $1$, the probability that a sample falls within $\mathcal{X}_N$ improves to the ratio $|\mathcal{S}_j|/(2L\sqrt{n})$, where $|\mathcal{S}_j| = |\alpha^*_+ - \alpha^*_{-}|$. Thus, the expected number of required MPC queries is bounded by $(2L\sqrt{n})/|\mathcal{S}_j|$. While this avoids the exponential scaling of uniform volumetric rejection sampling, the expected number of required MPC queries remains highly sensitive to the geometry of the constraints $\mathbb{X}$ and the feasible set $\mathcal{X}_N$, and is particularly high when $\mathcal{X}_N$ is thin relative to $\mathbb{X}$.\par
\textit{Bisection-Based HR:} To mitigate the computational risk associated with thin feasible segments, geometric bisection replaces probabilistic sampling with a deterministic search. Specifically, it approximates the boundary of $\mathcal{X}_N$ along a search line by maintaining a bracketing interval. Initialized with a search interval of worst-case width $(2L\sqrt{n})$, the algorithm iteratively verifies the feasibility of the interval's midpoint by solving the MPC. To reduce the bracket width to a precision $\epsilon>0$, the method requires $\lceil \log_2((L\sqrt{n})/\epsilon) \rceil$ iterations, which, consequently, constitutes the number of required MPC queries. Especially under high-precision requirements, this number can still be large compared to the fixed number of two LP queries required for the proposed method.\par 
In summary, casting the boundary identification as a linear program reduces the required optimization queries to a strict constant of two. This contrasts with the baseline methods, whose worst-case number of required MPC queries all scale with the state dimension $n_x$ according to $L^{n_x}$, $(2L\sqrt{n_x})/|\mathcal{S}_j|$, and $\log_2((2L\sqrt{n_x})/\epsilon)$. Consequently, the proposed method achieves a significant algorithmic speedup, proving particularly advantageous in high-precision settings or when the feasible set is small relative to the state constraints.\par
Finally, we acknowledge that problem~\eqref{eq: lp_oracle_final} has one additional variable $\alpha$ relative to~\eqref{eq: linear MPC}, incurring a small, often negligible, overhead. Conversely, Algorithm~\ref{alg: lmpc_hr} evaluates the LP only within the feasible set, whereas baseline methods solve the MPC repeatedly for potentially infeasible states, which can require many solver iterations to detect infeasibility.\par
\section{Numerical Experiments}\label{sec: experiments}
This section demonstrates the efficiency of our proposed approach on a linearized multi-rod inverted pendulum, comparing the cost of feasible state generation against the baseline methods introduced in Section~\ref{sec: lin complexity discussion}. To evaluate the impact of dimensionality, we vary the number of rods from one to three and assess the sampling efficiency across these increasing state dimensions.\par
\subsection{Example System and MPC Formulation}
We consider a chain of $n$ rigid rods connected by revolute joints, with the control objective of balancing all rods in the upright position. Each rod $i$ has length $l_i=1 \, \mathrm{m}$ and a point mass $m_i=1 \, \mathrm{kg}$ at its tip, and is actuated by a joint torque $\tau_i$. The state  is defined as $x = (\theta_1, \ldots, \theta_n, \dot{\theta}_1, \ldots, \dot{\theta}_n) \in \mathbb{R}^{2n}$, consisting of the absolute angles and angular velocities of each rod with respect to the vertical axis. The input $u = (\tau_1, \ldots, \tau_n) \in \mathbb{R}^{n}$ is defined by the applied joint torques. See Figure~\ref{fig:system} for a schematic. The discrete-time system dynamics are given by
\begin{equation}
    x(k+1) = \begin{bmatrix} I & \Delta t\, I \\ \Delta t\, M^{-1} D & I \end{bmatrix} x(k) + \Delta t \begin{bmatrix} 0 \\ M^{-1} \end{bmatrix} u(k).
\end{equation}
Here, $M \in \mathbb{R}^{n \times n}$ is the mass matrix with each element defined by $[M]_{ij} = n + 1 - \max(i,j)$, $D = g\,\mathrm{diag}(n{+}1{-}i) \in \mathbb{R}^{n\times n}$ is the gravity matrix, with $g = 9.81\,\mathrm{m/s^2}$, $I \in \mathbb{R}^{n \times n}$ is the identity matrix, and $\Delta t=0.1\, \mathrm{s}$ is the sampling time.
We formulate the MPC problem~\eqref{eq: linear MPC} with the costs $\ell(x,u) = \lVert x \rVert^2 + \lVert u \rVert^2$, $V_\mathrm{f}=0$, prediction horizon $N = 15$, and constraints $ \mathbb{X}=\{x \mid|\theta_i| \leq 2.5, |\dot\theta_i| \leq 3.5, i=1,\ldots, n\}$, $\mathbb{U}=\{u \mid |\tau_i| \leq 2, i=1,\ldots,n\}$, and $\mathbb{X}_\mathrm{f}=\{0\}$.\par 
\begin{figure}[ht] 
    \centering
    \includegraphics[]{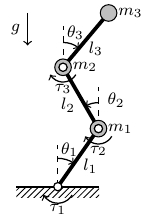}
    \caption{Schematic of the multi-rod inverted pendulum.}
    \label{fig:system}
\end{figure}
\subsection{Methodology}
The goal is to generate a dataset $\mathcal{D}=\{s(x_j)\}_{j=1}^{N_s}$ of $N_s = 1000$ uniformly distributed feasible samples $s(x_j) =x_j \in \mathcal{X}_N$ for each configuration $n \in \{1,2,3\}$. To this end, we implement the following strategies:
\begin{enumerate}
    \item \textit{Uniform Volumetric Rejection Sampling (UVRS)}: Samples are drawn uniformly from $\mathbb{X}$. For each sample, we try to solve~\eqref{eq: linear MPC} and add it to $\mathcal{D}$ if it is feasible. If it is infeasible, we reject and draw a new sample.
    \item \textit{Directional Rejection Sampling-Based HR (DRS-HR)}: The standard HR algorithm described in Section~\ref{sec: problem formulation} is deployed, where in each iteration $j$ state samples are drawn uniformly from the segment $\bar{S}_j = \{x_j\pm \alpha d_j \mid \alpha \geq 0\} \cap \mathbb{X}$ and rejected if infeasibility is detected by the MPC solver.
    \item \textit{Bisection-Based HR (BS-HR)}: Similarly to DRS-HR, the standard HR algorithm described in Section~\ref{sec: problem formulation} is deployed. However, for each sample $x_j$, the feasible line segment $S_j$ is identified with accuracy $\epsilon=0.001$ by geometric bisection based on the bounding box $\mathbb{X}$. 
    \item \textit{LMPC-HR}: Samples are drawn by the proposed LMPC-HR sampler detailed in Algorithm~\ref{alg: lmpc_hr}.
\end{enumerate}
All computations are performed on an AMD EPYC $7643$ compute node ($3.64\,$GHz, $1$ core, $2\,$GB RAM).
The optimization problems~\eqref{eq: linear MPC} and~\eqref{eq: lp_oracle_final} are solved using IPOPT with the CasADi interface \citep{Casadi_Andersson}.
\subsection{Results and Discussion}
\begin{table*}[t]
\centering
\caption{Comparison of sampling methods for generating $N_s=1000$ feasible samples for the multi-rod inverted pendulum with $n \in \{1,2,3\}$ rods. Entries marked ``--'' indicate that the method did not terminate within the allotted computation time.}
\label{tab:comparison}
\begin{tabular}{l cc cc cc}
\toprule
 & \multicolumn{2}{c}{$n=1\: (n_x=2)$} & \multicolumn{2}{c}{$n=2 \: (n_x=4) $} & \multicolumn{2}{c}{$n=3 \: (n_x=6)$} \\
\cmidrule(lr){2-3} \cmidrule(lr){4-5} \cmidrule(lr){6-7}
\textbf{Method} & \textbf{Solver Queries} & \textbf{Cost/Sample} & \textbf{Solver Queries} & \textbf{Cost/Sample} & \textbf{Solver Queries} & \textbf{Cost/Sample} \\
\midrule
UVRS   & $12380$ & $12.4$ & -- & -- & -- & -- \\
DRS-HR & $6311$  & $6.3$  & $17067$ & $17.1$ & $23507$ & $23.5$ \\
BS-HR  & $23485$ & $23.5$ & $23425$ & $23.4$ & $23362$ & $23.4$ \\
\textbf{LMPC-HR} & $\mathbf{2000}$ & $\mathbf{2.0}$ & $\mathbf{2000}$ & $\mathbf{2.0}$ & $\mathbf{2000}$ & $\mathbf{2.0}$ \\
\bottomrule
\end{tabular}
\end{table*}
Table~\ref{tab:comparison} compares the efficiency of each method in terms of total solver queries and the effective average cost per sample. For the baseline methods, each query involves solving the MPC problem~\eqref{eq: linear MPC}, which for $n=1$, $n=2$, and $n=3$, respectively, takes on average $0.02 \, \mathrm{s}$, $0.07 \, \mathrm{s}$, and $0.2 \, \mathrm{s}$ for a feasible state and $0.04 \, \mathrm{s}$, $0.23 \, \mathrm{s}$, and $0.68 \, \mathrm{s}$ for an infeasible one. For our LMPC-HR sampler, a query refers to solving problem~\eqref{eq: lp_oracle_final}, which for $n=1$, $n=2$, and $n=3$, respectively, takes on average $0.02 \, \mathrm{s}$, $0.05 \, \mathrm{s}$, and $0.15 \, \mathrm{s}$.\par
The UVRS approach demonstrates the expected poor scalability. While its efficiency is reasonable for the single-rod case ($n=1$) with an average cost of $12.4$ per feasible sample, it failed for higher dimensions. For $n \ge 2$, the method could not generate the target $1000$ samples within an 8-hour limit. At that point, it obtained $465$ feasible samples from a total of $117653$ attempts for $n=2$, and only $3$ feasible samples from $37455$ attempts for $n=3$.
The DRS-HR approach yields the lowest cost per sample among the baseline methods, requiring $6.3$ queries for $n=1$, which corresponds to a rejection rate of approximately $84\%$. Consistent with the analysis in Section~\ref{sec: lin complexity discussion}, the rejection rate worsens as the dimension increases, causing the average sampling cost to rise to $17.1$ for $n=2$ and $23.5$ for $n=3$.
The BS-HR approach exhibits an empirically constant sampling cost, contrasting with UVRS and DRS-HR. This reveals that, for the evaluated dimension increments, the dependency on $\log_2(\sqrt{n_x})$ is negligible compared to the dominant constant term $\log_2(2L/\epsilon)$. Nevertheless, the cost is high, remaining at approximately 23.4 solver queries per feasible sample.
In contrast, the proposed LMPC-HR sampler achieves a low and constant cost of exactly two LP queries per feasible sample across all dimensions, significantly outperforming the baseline methods. Furthermore, these queries are significantly faster, especially in higher dimensions, mainly because problem~\eqref{eq: lp_oracle_final} is only evaluated at feasible points. This combination of fewer and faster queries gives the LMPC-HR sampler a substantial advantage in overall sampling time.
%
%
%
%
%
%
%
%
%
\section{Conclusion}\label{sec: conclusion}
This paper presented the LMPC-HR sampler for efficient and approximately uniform sampling from the feasible set for approximate linear MPC. 
The core contribution is the formulation of the boundary identification problem within the HR algorithm as an optimization problem, replacing iterative feasibility checks with a single optimization step. For linear MPC with polyhedral constraints, this optimization problem is a convex LP, enabling exact and fast boundary identification. A case study on a multi-rod inverted pendulum with up to six states demonstrated an order of magnitude improvement in sampling efficiency over naive baselines, with the advantage growing in higher dimensions.

\bibliography{references}             

\end{document}